\documentclass[a4paper,onecolumn,accepted=2026-06-05]{quantumarticle}
\pdfoutput=1

\usepackage[T1]{fontenc}
\usepackage[british]{babel}
\usepackage{csquotes}
\usepackage{indentfirst}

\usepackage[dvipsnames]{xcolor}
\usepackage[unicode,colorlinks=true,citecolor=purple,linkcolor=purple,urlcolor=purple, breaklinks=true]{hyperref}

\usepackage[backend=biber,style=phys,biblabel=brackets,eprint=true,maxnames=10,giveninits=true, doi=false, url=false]{biblatex}
\DeclareFieldFormat{linkedtitle}{\href{\thefield{url}}{#1}}
\AtEveryBibitem{%
  \ifentrytype{book}{%
    \iffieldundef{url}{%
    }{%
      \DeclareFieldFormat[book]{title}{\color{purple}\mkbibquote{\ifhyperref{\DeclareFieldFormat[book]{title}{\addtocapture{linkedtitle}#1}}{#1}}}%
      
      \DeclareFieldFormat[book]{title}{\href{\thefield{url}}{\mkbibitalic{#1}}}
    }%
  }{}%
}
\usepackage{tikz-cd}
\usetikzlibrary{decorations.pathreplacing}
\usepackage{graphicx}
\usepackage{bm}

\usepackage{amsmath}
\usepackage{amsthm}
\newtheorem{theorem}{Theorem}
\newtheorem{lemma}[theorem]{Lemma}
\newtheorem{prop}[theorem]{Proposition}
\newtheorem{corollary}[theorem]{Corollary}

\usepackage{amssymb}

\newcommand{\ket}[1]{\left|#1\right\rangle}
\newcommand{\bra}[1]{\left\langle#1\right|}

\newcommand{\expect}[1]{\langle#1\rangle}

\newcommand{\tr}[1]{\mathrm{tr}\left(#1\right)}
\newcommand{\spec}[1]{\mathrm{spec}\left(#1\right)}
\newcommand{\um}{\mathfrak{u}(m)}
\newcommand{\imum}{d\varphi(\mathfrak{u}(m))}
\newcommand{\imsum}{d\varphi(\mathfrak{su}(m))}
\newcommand{\uM}{\mathfrak{u}(M)}
\newcommand{\Um}{\mathrm{U}(m)}
\newcommand{\UM}{\mathrm{U}(M)}
\newcommand{\imUm}{\varphi(\mathrm{U}(m))}
\newcommand{\dphi}{d\varphi}
\newcommand{\dphiinv}{d\varphi^{-1}}
\newcommand{\udag}{U^\dagger}

\begin{document}

\title{Lie algebraic invariants in quantum linear optics}

\author[1,4]{Pablo V. Parellada}
\email{pablo.veganzones@uva.es}
\author[2]{Vicent Gimeno i Garcia}%
\author[2]{Julio Jos\'e Moyano-Fern\'andez}
\author[3,4]{Juan Carlos Garcia-Escartin}

\affil[1]{Universidad de Valladolid, Dpto. F\'isica Te\'orica, At\'omica y \'Optica, 47011 Valladolid, Spain.}
\affil[2]{%
Universitat Jaume I, Dpto.~de Matem\`atiques \& IMAC, E-12071, Castell\'{o}, Spain.
}%
\affil[3]{Universidad de Valladolid, Dpto.~Teor\'ia de la Se\~{n}al e Ing.~Telem\'atica, Paseo Bel\'en n$^o$ 15, 47011 Valladolid, Spain.}
\affil[4]{Laboratory for Disruptive Interdisciplinary Science (LaDIS), Universidad de Valladolid, 47011 Valladolid, Spain.}

\date{5 June 2026}

\begin{abstract}
Quantum linear optics without post-selection is not powerful enough to produce any quantum state from a given input state. This limits its utility since some applications require entangled resources that are difficult to prepare. Thus, a deeper understanding of linear optical state preparation is needed. In this work, we give a recipe to derive conserved quantities in the evolution of arbitrary states along any possible passive linear interferometer. One example of such an invariant is the spectrum of a density matrix mapped onto the Lie algebra of passive linear optical Hamiltonians. These invariants give necessary conditions for exact state preparation: if the input and output states have different invariants, it is impossible to design a passive linear interferometer that evolves one into the other. Moreover, we provide a lower bound to the distance between an output and target state based on the distance between their invariants. This gives a necessary condition for approximate or heralded state preparations. Therefore, the invariants allow us to narrow the search when trying to prepare useful entangled states, like NOON states, from easy-to-prepare states, like Fock states. We conclude that future exact and approximate state preparation methods will need to consider the necessary conditions given by our invariants to weed out impossible linear optical evolutions.

\end{abstract}
\maketitle

\section{Introduction}

Passive linear interferometers, like the Michelson interferometer, are simple devices that can be readily built in the lab and have ubiquitous applications in physics. They consist of a set of modes in which light can propagate (distinct spatial paths, orthogonal polarizations...) and some elements that create an interaction between the light in the modes (beam splitters and phase shifters).

Even though linear interferometers arose as classical optical devices, they can be studied from the quantum formalism, in which they exhibit surprising properties like the Hong-Ou-Mandel effect \cite{hong_measurement_1987}. Quantum linear optics has brought lots of applications in quantum metrology \cite{barbieri_optical_2022}, quantum communication \cite{duan_long-distance_2001} and quantum computation \cite{kok_linear_2007}.

However, linear optics alone is not powerful enough to perform universal quantum computation, since it cannot implement all unitary gates — except for an exponential number of modes and one photon \cite{cerf_optical_1998, moyano-fernandez_linear_2017}. To perform optical quantum computing, some kind of non-linearity is generally needed, like measurement \cite{knill_scheme_2001}, or exotic resources, like GKP states \cite{konno_logical_2024}. This might seem a bit disappointing, but Aaronson and Arkhipov proved that linear optics alone can solve some problems — namely \textit{boson sampling} — much faster than classical computers \cite{AA13}, giving rise to the first ``quantum advantage'' that could be tested experimentally \cite{brod_photonic_2019}. This result positioned quantum linear optical devices as an intermediate candidate between a classical computer and a fully-fledged quantum computer (see, for example, photonic integrated circuits \cite{PMS14,EPS20,bogaerts_programmable_2020,TMS21,HPG22,TAdG23}).

As we have said, not all quantum evolutions can be implemented with passive linear optics. This raises the question: given two arbitrary photonic states, is there a linear optical evolution connecting them? This state preparation problem is especially interesting since many useful applications of linear optics begin with some entangled input state. While giving a sufficient condition is tough and has not been answered yet in the general case (for two modes or two photons see \cite{MRO14, Entanglement-2001}), there are necessary conditions that must be fulfilled. The simplest one is the conservation of the mean number of photons: if the states have different mean values, the transition is impossible. Other invariants were previously proposed for a fixed number of photons \cite{MRO14,PGM23}.

In this article, we use the Lie algebras describing passive linear optics to define other quantities that must be conserved if there is a linear optical evolution connecting two states. That is, each state can be assigned different invariants and, if any of them do not coincide, there cannot exist a passive interferometer connecting them. Additionally, we show that for many of those invariants, if we want to reach a state close in trace distance to a given target, the input and target states must also have close invariants.

Some of the invariants we propose can be measured experimentally and have been recently measured by two independent groups \cite{rodari_observation_2025, yang_experimental_2025}, with a possible application to the verification of boson sampling mentioned in \cite{rodari_observation_2025}.

Readers not particularly interested in linear optics can also find useful mathematical tools in this article. For example, the adjoint action and the mapping of quantum states onto subalgebras — which we use to derive the invariants — have been used recently to study barren plateaus in variational quantum circuits \cite{fontana_characterizing_2024, ragone_lie_2024}.

The article is organized as follows. In Section \ref{sec: Mathematical preliminaries}, we establish the mathematical description of passive linear optics, which is necessary to derive the invariants. In Section \ref{sec: invariants}, we propose several invariant quantities under linear optical evolution. In Section \ref{sec: examples}, we provide examples of application of the invariants. In Section \ref{sec: approximate state preparation}, we introduce the notion of distance between invariants and relate it to the distance between quantum states, with a possible application to studying heralded preparations. In Section \ref{sec: code}, we discuss the software implementation of the invariants. Finally, in Section \ref{sec: discussion}, we draw some conclusions and outline future lines of research.

\section{Mathematical preliminaries}\label{sec: Mathematical preliminaries}

\subsection{Lie groups}\label{sec: lie groups}

A matrix Lie group is a group of matrices that can be parametrized smoothly by real numbers with the property that the product and inverse of matrices are smooth functions of these parameters. One example is the Lie group of $m\times m$ unitary matrices, which is denoted by $\Um$.

In classical linear optics, unitary matrices describe the evolution (or scattering) of the electromagnetic field through a passive interferometer. If the interferometer can propagate $m$ modes of the electromagnetic field, the evolution (or scattering) of the electromagnetic field amplitudes in each mode is described by a unitary matrix in $\Um$ called the \textit{scattering matrix}. One possible parametrization of scattering matrices is via the angles of the beam splitters and phase shifters that make up the interferometer \cite{reck_experimental_1994}.

In quantum optics, the scattering matrix describes the evolution a single photon; that is, the evolution of the creation, $a_i^\dagger$, and annihilation, $a_i$, operators \cite{aniello_exploring_2006}:
\begin{equation}
     a_i^\dagger \rightarrow \sum_{j=1}^m S_{ij}^\ast a_j^\dagger
     \:, \qquad 
     a_i \rightarrow \sum_{j=1}^m S_{ij} a_j
     \:.
\end{equation}
The unitarity of the scattering matrix $S$ implies that passive linear optics conserves the total number of photons.

The evolution of a quantum state with multiple photons along an interferometer is given by a unitary operator $U$ acting on Fock space, $\mathcal{F}$. This operator acts on the Fock basis, $\ket{n_1 \ldots n_m}$, by evolving the creation operator in each mode with $S$ \cite{skaar_quantum_2004}:
\begin{equation}\label{eq: heisenberg evolution}
    U \ket{n_1 \ldots n_m} = \prod_{k=1}^m \frac{1}{\sqrt{n_k!}}\left(\sum_{j=1}^m S_{jk}a_j^\dagger\right)^{n_k} \ket{0\ldots 0} \:.
\end{equation}
(An alternative way to compute $U$ from $S$ is using the permanents \cite{scheel_permanents_2004}). This relation between $S$ and $U$ is given by a unitary representation of $U(m)$ on Fock space \cite{aniello_exploring_2006}:
\begin{equation}\label{eq: group homomorphism fock space}
    \begin{aligned}
        \varphi: \:\: \mathrm{U}(m) \:&\rightarrow\: \mathrm{U}(\mathcal{F}) \\
        S \: &\mapsto\: U = \varphi(S)  \:,
    \end{aligned}
\end{equation}
where $\mathrm{U}(\mathcal{F})$ is the group of unitary operators acting on Fock space. A representation is, by definition, a group homomorphism: a map between groups such that
\begin{equation}
    \varphi(S_2 S_1) = \varphi(S_2)\varphi(S_1) \:.
\end{equation}
Physically, this means that evolving a quantum state through an interferometer with scattering matrix $S_2 S_1$ is equivalent to evolving it first through interferometer $S_1$ and then through interferometer $S_2$.

Since linear optics preserves the total number of photons, this representation is known to be irreducible (has no non-trivial invariant subspaces) in each subspace of $n$ photons in $m$ modes \cite{aniello_exploring_2006}. This Hilbert space, denoted by $\mathcal{H}_{n,m}$, is spanned by Fock states $\ket{n_1 \ldots n_m}$ such that with $n_1 + \cdots + n_m = n$. Hence, its dimension is
$$
\dim_{\mathbb{C}} \mathcal{H}_{n,m} = \binom{m+n-1}{n} \equiv M \:,
$$
which corresponds to the number of ways one can distribute $n$ photons in $m$ modes. In summary, the infinite dimensional Fock space can be decomposed into finite dimensional invariant subspaces, $\mathcal{F} = \bigoplus_{n=0}^\infty \mathcal{H}_{n,m}$. When restricted to $\mathcal{H}_{n,m}$, the representation \eqref{eq: group homomorphism fock space} is a map between finite dimensional unitary groups
\begin{equation}\label{eq: group homomorphism}
    \begin{aligned}
        \varphi: \:\: \mathrm{U}(m) \:&\rightarrow\: \mathrm{U}(M) \\
        S \: &\mapsto\: U = \varphi(S)  \:,
    \end{aligned}
\end{equation}
where $\UM$ is the group of unitary matrices that describes the evolution of quantum states in $\mathcal{H}_{n,m}$. This map is not surjective because not all quantum evolutions can be realized with linear optics \cite{moyano-fernandez_linear_2017}. Thus, linear optical unitaries form a subgroup of $\UM$, denoted by $\imUm$, with the same dimension as $\Um$.

\subsection{Lie algebras}

The Lie algebra $\mathfrak{g}$ of a matrix Lie group $G$ is the tangent space at the identity; or, alternatively, the set of matrices $X$ such that $e^{X}$ is in $G$. Lie algebras are equipped with a \textit{Lie bracket}: a bilinear and antisymmetric map $[\cdot,\cdot ]: \mathfrak{g} \times \mathfrak{g} \rightarrow \mathfrak{g}$ that satisfies the Jacobi identity \cite{hall_lie_2015}.

For example, the Lie algebra of the unitary group is the vector space of anti-Hermitian matrices with the commutator as a Lie bracket. In fact, this algebra of anti-Hermitian matrices is isomorphic to the algebra of Hermitian matrices: just multiply an anti-Hermitian matrix by $-i$ (to get a Hermitian matrix) and multiply the commutator by $i$ (to get a Lie bracket that closes the algebra). The exponential that maps the algebra to the unitary group is also modified with an $i$ factor: $e^{iH}$ is unitary, where $H$ is Hermitian. For convenience, throughout this work we choose to define $\um$ as the Lie algebra of Hermitian matrices.

While the unitary group $\UM$ describes quantum evolution operators, its algebra, $\uM$, describes their corresponding Hamiltonians. In particular, the Hamiltonians of passive linear optical unitaries (a subgroup of $\uM$) are given by the operators $H = \sum_{j k} h_{j k} a_j^\dagger a_k$, where $h$ is Hermitian, restricted to the subspace of a fixed number of photons $\mathcal{H}_{n,m}$ \cite{leonhardt_explicit_2003, aniello_algebraic_2005, aniello_exploring_2006, GGM18}. In fact, the map from $h$ to $H$ is a Lie algebra homomorphism \cite{GGM18},
\begin{equation}\label{eq: algebra homomorphism}
    \begin{aligned}
        d\varphi: \:\: \um \:&\rightarrow\: \uM \\
        h \: &\mapsto\: H = \sum_{j k} h_{j k} a_j^\dagger a_k  \:,
    \end{aligned}
\end{equation}
which means that $d\varphi$ is a linear map such that $d\varphi(i[h_1, h_2]) = i[d\varphi(h_1), d\varphi(h_2)]$. The image of $\dphi$, denoted by $\imum$, is a subalgebra of $\uM$ and it is precisely the Lie algebra of the linear optical unitaries, $\varphi(\Um)$.

To understand the physical meaning of $h$, one can restrict the Hamiltonian to the subspace of a single photon. In this case, $H = \dphi(h) = h$ is a Hamiltonian for the unitary evolution of a single photon, which is given by $\varphi(S)=S$. Thus, $h$ is related to the scattering matrix of the interferometer by the exponential map: $S = e^{ih}$. We summarize the group and algebraic structure of passive linear optics with the following commutative diagram (see e.g. \cite{GGM18}):
\begin{center}
    \begin{tikzcd}
    h\in\mathfrak{u}(m) \arrow{r}{d \varphi}\arrow{d}[swap]{\exp} &  H =  \sum_{ij} {h_{ij}{a}_i^\dagger {a}_j} \arrow{d}{\exp} \\
    S=e^{ih}\in {\rm U}(m)\arrow{r}[swap]{ \varphi} & U = e^{i\sum{h_{ij}{a}_i^\dagger {a}_j}}
    \end{tikzcd} \:\:.
\end{center}

The invariants that we derive later on in this article arise from expectation values of a basis of linear optical Hamiltonians. However, not all bases are valid; as we shall see, we need one that is the image of an orthonormal basis of $\mathfrak{u}(m)$. An example of such a basis is given by (using the notation of \cite{rodari_observation_2025}):
\begin{align}
    &b^x_{jk}=\frac{1}{\sqrt{2}}(\ket{j}\bra{k}+\ket{k}\bra{j}) &{\rm for}\quad 1 \leq j < k \leq m, \nonumber \\
    &b^y_{jk}=\frac{i}{\sqrt{2}}(\ket{j}\bra{k}-\ket{k}\bra{j}) &{\rm for}\quad 1 \leq j < k \leq m, \label{eq: basis algebra} \\
    &b^z_j \:=\:\ket{j}\bra{j}  &{\rm for} \quad j= 1,\ldots , m. \quad \nonumber
\end{align}
We will call this basis $\{b_i\}$ for short. Taking its image under $\dphi$ we obtain a basis of $\imum$, which we denote by $\{O_i\}$:
\begin{align}
    &O_{jk}^x=\frac{1}{\sqrt{2}}\left( a^\dag_{j} a_k+ a^\dag_{k} a_j\right) &{\rm for}\quad 1 \leq j < k \leq m, \nonumber \\
    &O_{jk}^y=\frac{i}{\sqrt{2}}\left( a^\dag_{j} a_k- a^\dag_{k} a_j\right) &{\rm for}\quad 1 \leq j < k \leq m, \label{eq: basis image algebra} \\
    &O_j^z= n_j= a^\dag_{j} a_j  &{\rm for} \quad j= 1,\ldots , m, \quad \nonumber
\end{align}
where the operators can be restricted, if needed, to the subspace of a fixed number of photons, $\mathcal{H}_{n,m}$. In the special case of having just one photon, $d\varphi$ is the identity map and $\{O_i\}$ is the same basis as $\{b_i\}$. This gives a physical interpretation of $\{b_i\}$ as a basis of single-photon Hamiltonians.

This basis is especially interesting because the expectation values of these operators can be measured experimentally. The operators correspond to photon number and angular momentum observables \cite{YMK86,campos_quantum-mechanical_1989} and measurement can be done with photon counting and a homodyne detection setup \cite{NK06}. First, measuring $a^\dag_{i} a_i$ corresponds to counting the number of photons in the $i$-th mode (although experimentally it is not trivial). Then, operators $(a^\dag_{j} a_k+ a^\dag_{k} a_j)/\sqrt{2}$ correspond to a 50:50 beam splitter in modes $j$ and $k$, and measuring $n_j - n_k$ after. Finally, measuring $i(a^\dag_{j} a_k - a^\dag_{k} a_j)/\sqrt{2}$ is the same as measuring the previous operator, but adding a $\pi/2$ phase shifter in mode $k$ before the beam splitter.

\subsection{Adjoint action}

Given a Lie group $G$, the conjugation by an element $g\in G$ is an automorphism in $G$: $C_g(h) = ghg^{-1} \in G \:.$
If $\mathfrak{g}$ is the Lie algebra of $G$, the differential of $C_g$ at the identity is a map in $\mathfrak{g}$ called the adjoint map: ${\rm Ad}_g: \mathfrak{g} \mapsto \mathfrak{g}$. For matrix Lie groups, it is given by \cite[Definition 3.32]{hall_lie_2015}:
\begin{equation}
    {\rm Ad}_{g}(X) = g X g^{-1} \:,
\end{equation}
where $g\in G$ and $X \in \mathfrak{g}$. This map defines an action of the Lie group $G$ over its Lie algebra $\mathfrak{g}$ called the adjoint action, ${\rm Ad}: G \times \mathfrak{g} \rightarrow \mathfrak{g}$. In the Schrödinger picture of quantum mechanics, one finds this action describing the evolution of a density matrix $\rho$ by a unitary $U$:
\begin{equation}
    \rho \mapsto U \rho U^\dag \:,
\end{equation}
so it is natural to study this action to find invariants under linear optical evolution. In the Heisenberg picture, where observables—rather than states—evolve with a unitary operator, the adjoint action also describes the evolution of these observables:
\begin{equation}
    A \mapsto U^\dag A U \:.
\end{equation}
(Notice that $U$ and $U^\dag$ are interchanged with respect to the Schrödinger picture).

Let us describe first the adjoint action of a single-photon unitary, $S\in \Um$, acting on the basis of single-photon Hamiltonians \eqref{eq: basis algebra},  $\{b_i\}\in \um$:
\begin{equation}\label{eq: adjoint action}
    S^\dag b_i S = \sum_{j=1}^{m^2}C_{ij} b_j \:,
\end{equation}
where $C_{ij} \in \mathbb{R}$. In fact, since $\{b_i\}$ is an orthonormal basis, $C$ must be an orthogonal matrix:
\begin{equation}
    \delta_{ij} = 
    {\rm tr}\left(b_i b_j\right)= 
    {\rm tr}\left(S^\dag b_i S S^\dag b_j S\right) = 
    {\rm tr}\left(\sum_{i'j'}C_{ii'}C_{jj'} b_{i'}b_{j'} \right) = 
    \sum_{i'}C_{ii'}C_{ji'} = (C C^T)_{ij} \:.
\end{equation}

To obtain the adjoint action of a linear optical unitary, $U=\varphi(S) \in  \Um$, on the basis of multiphoton Hamiltonians \eqref{eq: basis image algebra}, $\{O_i\} \in \imum$, we need the following lemma:
\begin{lemma}\label{lemma: adjoint action}
    Let $G$ and $H$ be matrix Lie groups with Lie algebras $\mathfrak{g}$ and $\mathfrak{h}$. Let
    $\varphi:G\rightarrow H$ be a Lie group homomorphism and $d\varphi:\mathfrak{g}\rightarrow \mathfrak{h}$ its differential at the identity (a Lie algebra homomorphism). Then, for every $g\in G$ and $X\in \mathfrak{g}$,
    \begin{equation}
        d\varphi(g^{-1} X g) = \varphi(g)^{-1} d\varphi(X) \varphi(g)  \:.
    \end{equation}
\end{lemma}
\begin{proof}
    This can be proved by considering a curve passing through the identity at $t=0$ with velocity $g^{-1} X g$, namely $\gamma(t) = e^{tg^{-1}Xg} = g^{-1} e^{tX} g$. By definition of $d\varphi$, taking the derivative of the image curve $\varphi(\gamma(t))$ at $t=0$ yields:
    \begin{equation}
        \begin{aligned}
        d\varphi(g^{-1} X g)
        &=\left. \frac{d}{dt} \varphi(e^{tg^{-1} Xg}) \right|_{t=0} \\[1mm]
        &=\left. \frac{d}{dt} \varphi(g^{-1} e^{tX} g) \right|_{t=0} \\[1mm]
        &= \left. \frac{d}{dt} \varphi(g)^{-1} \varphi(e^{tX}) \varphi(g) \right|_{t=0} \\[1mm]
        &= \varphi(g)^{-1} d\varphi(X) \varphi(g) \:,
        \end{aligned}
    \end{equation}
where in the third equality we used that $\varphi$ is a group homomorphism and in the last equality we used the definition of $d\varphi$ as the differential of $\varphi$ at the identity.    
\end{proof}

Applying Lemma \ref{lemma: adjoint action} to the photonic homomorphism \eqref{eq: group homomorphism} and the adjoint action \eqref{eq: adjoint action}, we conclude that a multiphoton linear optical unitary acting on the basis of multiphoton Hamiltonians \eqref{eq: basis image algebra} is completely determined by a single-photon unitary (scattering matrix) acting on the single-photon Hamiltonians \eqref{eq: basis algebra}:
\begin{equation}\label{eq: adjoint action image}
    U^\dag O_i U = \varphi(S^\dag)\, d\varphi(b_i)\, \varphi(S) = d\varphi\left(S^\dag b_i S \right) = d\varphi\left(\sum_{j=1}^{m^2}C_{ij}b_j\right) = \sum_{j=1}^{m^2}C_{ij}d\varphi(b_j) = \sum_{j=1}^{m^2}C_{ij} O_j \:,
\end{equation}
where $C$ is the same orthogonal matrix from Eq. \eqref{eq: adjoint action}. Similarly, $U O_i U^\dag = \sum_{j=1}^{m^2}C_{ij}^T O_j$. In fact, Eq. \eqref{eq: adjoint action image} (formally) holds even if the $O_i$ are infinite dimensional (block diagonal) operators acting on the full Fock space, $\bigoplus_{n=0}^\infty \mathcal{H}_{n,m}$. Equation \eqref{eq: adjoint action image}  will be very important later on when deriving our invariants, so we collect it in a lemma:
\begin{lemma}\label{lemma}
If $\{b_i\}$ is an orthonormal basis of $\mathfrak{u}(m)$, then letting $\{O_i = \dphi(b_i)\}$ be the associated basis of $\imum \subset \uM$, we have, for any linear optical unitary $U$: ${\rm Ad}_{U^\dag}(O_i)=\sum_{j=1}^{m^2}C_{ij} O_j$, where $C$ is an orthogonal matrix.
\end{lemma}

\subsection{Irreducible representations}\label{sec: irreducible representations}

In the previous section, we defined the adjoint action of $\Um$ on $\imum \subset \uM$ via the group homomorphism $\varphi$. In fact, this action can also be defined on the entire $\uM$:
\begin{equation}\label{eq: linear optical ad rep}
    \mathrm{Ad}_{\varphi(S)} (\rho) = \varphi(S) \,\rho \, \varphi(S)^\dagger \:, \qquad S \in \Um\:, \:\:\: \rho \in \uM \:.
\end{equation}
This is very convenient because the density matrix of a quantum state with $n$ photons and $m$ modes lives in $\uM$. This action defines a representation of $\Um$ on $\uM$ that, if vectorized, can be seen to be isomorphic to the tensor product of the representations $\varphi$, from Eq. \eqref{eq: group homomorphism}, and $\bar\varphi$, its complex conjugate representation:
\begin{equation}
    \mathrm{Ad}_{\varphi(S)} \cong \varphi (S) \otimes \bar \varphi (S) \:.
\end{equation}

Contrary to the irreducibility of the representation $\varphi$ on $\mathcal{H}_{n,m}$ (see Section \ref{sec: lie groups}), the representation $\mathrm{Ad}_{\varphi(S)}$ is not irreducible\footnote{For the basic concepts of representation theory, refer to Appendix \ref{appendix: invariant subspaces}.}. It was shown in \cite{arienzo_bosonic_2025, wilkens_benchmarking_2024} that this representation can be decomposed into a trivial representation and $n$ non-trivial inequivalent irreducible representations. The dimensions of the carrier spaces of the representation, $W_\ell$, are
\begin{equation}
    \mathrm{dim} \,W_{\ell} = \frac{m + 2 \ell - 1}{m - 1} \binom{m + \ell - 2}{\ell}^{2} \:, \qquad \ell = 0, 1, \ldots , n \:.
\end{equation}
The subspace $W_0$ corresponds to the trivial subspace generated by the identity, $\mathcal{I}$. The subspace $W_1$, with dimension $m^2-1$, is $\imsum$\footnote{Since the dimension of each $W_\ell$ is different, there are no isomorphic irreducible subspaces and each $W_\ell$ coincides with its isotypic component. By the uniqueness of the isotypic (or canonical) decomposition \cite[Sections 2.6, 4.3]{serre_representations}, we conclude that $W_1 = \imsum$, because both spaces are irreducible and share the same dimension.}, the image of $\mathfrak{su}(m)$ under $d\varphi$ (recall that $\mathfrak{su}(m)$ is the subalgebra of $\um$ consisting of traceless matrices).

In this work, we call $\imum = \mathcal{I} \oplus \imsum$ the \textit{tangent subspace} and its complement, $\imum^\perp$, the \textit{perpendicular subspace}, which could be further decomposed into invariant subspaces: $\imum^\perp = \bigoplus_{\ell=2}^{n+1} W_\ell$. The decomposition $\uM = \imum \oplus \imum^\perp$ will be useful in the following Section because it allows us to decompose a density matrix $\rho\in\uM$ into the tangent component of $\uM$ and everything else. Of course, we could refine this partition even further and decompose $\rho$ into $n+1$ components belonging to each $W_\ell$.

\section{Invariants}\label{sec: invariants}

\subsection{Tangent invariant}\label{sec: tangent invariant}

The density matrix $\rho$ of a state with $n$ total photons in $m$ modes is a Hermitian matrix of size $M \times M$, so it belongs to $\uM$. We consider mapping $\rho$ onto the subalgebra $\imum$ of multiphoton Hamiltonians:
\begin{equation}\label{eq: tangent projection}
    \rho_T = \sum_{i=1}^{m^2} \tr{O_i \rho} O_i \:,
\end{equation}
where $\{O_i\}$ is the basis \eqref{eq: basis image algebra} of $\imum$. We will call $\rho \mapsto \rho_T$ the \textit{tangent map}. Note that this map was also studied in \cite{somma_quantum_2005} in the context of entanglement and the norm of $\rho_T$ was shown to be invariant when $\{O_i\}$ is an orthonormal basis.

An interesting property of $\rho_T$ is that the evolution of $\rho$ under a linear optical $U$ commutes with the tangent map:
\begin{equation}\label{eq: evolution of the projection}
    \begin{aligned}
    &(U\rho U^\dagger)_T = \sum_{i=1}^{m^2}  {\rm tr}(O_i U\rho U^\dagger ) O_i = \sum_{i=1}^{m^2}  {\rm tr}(U^\dagger O_i U\rho) O_i \\
    &=\sum_{ij}C_{ij} {\rm tr}(O_j \rho ) O_i = \sum_{j}{\rm tr}(O_j \rho ) U O_j U^\dagger = U\rho_T U^\dagger \:.  \end{aligned}
\end{equation}
This property, known as equivariance, leads to our first invariant: the spectrum of $\rho_T$ is conserved. An important aspect of Eq. \eqref{eq: evolution of the projection} is that we only used the properties of the adjoint action. We didn't need to orthonormalize the basis $\{O_i\}$, so the invariant works even if the $O_i$ are infinite dimensional.

As illustrated in Figure \ref{fig:inclusion_spec}, the conservation of $\spec{\rho_T}$ gives a necessary condition for the existence of a linear optical $U$ connecting $\rho$ with another state in Fock space.

\begin{figure}[ht]
\centering
\includegraphics[scale=0.8]{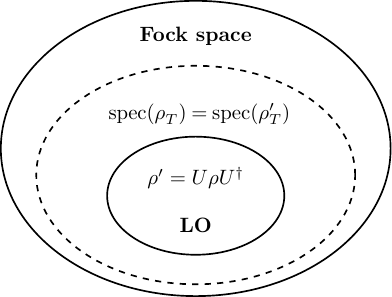}
\caption{For a given state $\rho$, there is a subset of states in Fock space that we can reach with linear optics (LO) and it is contained in the set of states sharing the same invariant.} \label{fig:inclusion_spec}
\end{figure}

Since $M$ grows very quickly with the number of photons, we have to diagonalize big matrices with a lot of redundant information to obtain the spectrum. To simplify this calculation, note that $\rho_T$ belongs to the algebra $\imum$, which is isomorphic to $\um$. Therefore,  we can take the inverse algebra homomorphism to map $\rho_T$ into a smaller matrix in $\um$:
\begin{equation}
    h_\rho := \dphiinv(\rho_T)=\sum_{i=1}^{m^2} \tr{O_i \rho} b_i \:,
\end{equation}
where $\{b_i\}$ is the orthonormal basis of $\um$ given in Eq. \eqref{eq: basis algebra}. We see that $h_\rho$ is also equivariant under the adjoint action of $S$:
\begin{equation} \label{eq: h_rho equivariance}
\begin{aligned}
    & h_{U\rho U^\dagger} = \sum_{i=1}^{m^2} \tr{\udag O_i U \rho} b_i = \sum_{ij}^{m^2}C_{ij} {\rm tr}(O_j \rho ) b_i = \\
    & = \sum_{ij}^{m^2}{\rm tr}(O_j \rho )C^T_{ji} b_i = \sum_{i=1}^{m^2} \tr{ O_i \rho} S b_i S^\dagger =  S h_\rho S^\dagger \:.
\end{aligned}
\end{equation}
Therefore, the spectrum of $h_\rho$ is invariant under linear optical evolution. (Recall that $S$ is the classical scattering matrix that describes the evolution of a single photon and it is related to $U$ via the photonic homomorphism $\varphi(S)=U$.)

We can summarize the previous paragraphs in the following theorem.

\begin{theorem}\label{thm: spectrum projection}
The spectrum of $\rho_T = \sum_{i=1}^{m^2}  {\rm tr}(O_i \rho) O_i$ is invariant under linear optical evolution. Moreover, this spectrum has the same information as the spectrum of $h_\rho = \sum_{i=1}^{m^2}  {\rm tr}(O_i \rho) b_i$.
\end{theorem}
\begin{proof}
    To prove that both invariants are equivalent, note that we can diagonalize $h_\rho$ with unitary matrices:
    $$h_\rho = V \Lambda V^\dagger \:.$$
    Then, applying the algebra isomorphism $d\varphi$ \eqref{eq: algebra homomorphism} and Lemma \ref{lemma: adjoint action}: 
    $$\rho_T = d\varphi(h_\rho) =
    d\varphi(V \Lambda V^\dagger) =
    \varphi(V)d\varphi(\Lambda)\varphi(V^\dagger) =
    \varphi(V)d\varphi(\Lambda)\varphi(V)^\dagger = 
    \varphi(V)D\varphi(V)^\dagger \:,
    $$
    where $\varphi(V)$ is unitary and $D = d\varphi(\Lambda)$ is diagonal, as one can check from the definition of $d\varphi$. Therefore, the spectrum of $h_\rho$ and the spectrum of $\rho_T$ are related by an isomorphism and must provide the same information. Namely, $\spec{\rho_T}=\spec{\rho_T'}$ if and only if $\spec{h_\rho}=\spec{h_{\rho'}}$.
    
\end{proof}

Let us analyze $h_\rho$ closely. The $j$-th diagonal element, corresponding to $b_j^z = \ket{j}\bra{j}$, is the mean number of photons in the $j$-th mode: $\tr{O_j^z \rho}=\expect{n_j}$. Meanwhile, the element $(j,k)$ is just $\expect{a_k^\dagger a_j}$. Thus, for two modes, the invariant is the spectrum of
\begin{equation}
    h_\rho = 
    \begin{pmatrix}
        \expect{n_1}  &  \expect{a_2^\dagger a_1}  \\
        \expect{a_1^\dagger a_2}  &  \expect{n_2}
    \end{pmatrix} \:.
\end{equation}
This invariant is identical to the invariant $ff^\dagger |_{k=1}$ in \cite{MRO14} (but with $\expect{n_j}$ instead of $\expect{a_j a_j^\dagger}$). In fact, $h_\rho$ is nothing but the quantum optical coherency matrix, which is also shown in \cite{rodari_observation_2025} to be equivalent to the spectrum of $\rho_T$.

Analytically computing the spectrum of $h_\rho$ can get complicated as the number of modes grows, so an alternative characterization would help. In Corollary \ref{corollary: trace invariants}, we follow the proof in \cite{Klein19631283} to show that the powers of $h_\rho$ are invariant under the adjoint action. This yields $m$ independent invariants, the maximum that can be built from the expectation values $\tr{O_i\rho}$, which follows from \cite{Klein19631283} noting that $\{\tr{O_i\rho}\}$ are the coordinates of $h_\rho$ in the basis $\{b_i\}$.

\begin{corollary}\label{corollary: trace invariants}
The following quantities are invariant under linear optical evolution of a quantum state $\rho$
\begin{equation}\label{eq: trace invariants}
    \left\lbrace
     \begin{aligned}
     I_1(\rho) &= {\rm tr}(n\rho) \:,\\
     I_2(\rho) &= \sum_{i=1}^{m^2}{\rm tr}(O_i\rho)^2\:,\\
     \vdots\\
     I_m(\rho) &= \sum_{i_1,\cdots,i_m} \tr{ b_{i_1} \cdots b_{i_m}}{\rm tr}( O_{i_1}\rho)\cdots{\rm tr}( O_{i_m}\rho)\:.
     \end{aligned}\right.
\end{equation}
Moreover, they are equivalent to the invariants of Theorem \ref{thm: spectrum projection}.
\end{corollary}
\begin{proof}
Each invariant just corresponds to taking traces of powers of $h_\rho$. Recalling Eq. \eqref{eq: h_rho equivariance}, we obtain:
\begin{equation}
    \begin{aligned}
    \tr{(h_{U \rho \udag})^k} = \tr{(S h_\rho S^\dag)^k} = \tr{h_\rho^k}=\tr{\sum_{i_1} \tr{O_{i_1}\rho}b_{i_1} \cdots \sum_{i_k}\tr{O_{i_k}\rho}b_{i_k}} \\
    = \sum_{i_1, \ldots, i_k}\tr{ b_{i_1} \cdots b_{i_k}}\tr{O_{i_1}\rho} \cdots \tr{O_{i_k}\rho} \:.
    \end{aligned}
\end{equation}
Since $\{b_i\}$ is an orthonormal basis, $\tr{b_i b_j} = \delta_{ij}$ and the invariant for $k=2$ turns out to be $\sum_i \tr{O_i \rho}^2$. For $k=1$, the invariant is simply the mean number of photons.

These $m$ invariants provide the same information as the spectrum. On the one hand, the spectrum of $h_\rho$ completely determines the traces of $h_\rho^k$. On the other hand, the characteristic polynomial of $h_\rho$ can be written, using the Leverrier-Faddeev algorithm \cite{Leverrier_Faddeev}, in terms of $\tr{h_\rho^k}$ for $k=1,\ldots m$.
\end{proof}

Note that the quantity $\tr{ b_{i_1} \cdots b_{i_k}}$ is a scalar that is invariant under the simultaneous transformation of the basis elements $\{b_i\}$ under the adjoint action:
\begin{equation}\label{eq: invariant tensor}
    \tr{ b_{i_1} \cdots b_{i_k}} =  \tr{ S^\dag b_{i_1} S \cdots S^\dag b_{i_k} S} = \sum_{j_1 \ldots j_k} C_{i_1 j_1} \cdots  C_{i_k j_k} \tr{b_{j_1} \cdots b_{j_k}} \:.
\end{equation}
In addition, this quantity also appears in the $m$ Casimir operators of the adjoint representation of $\Um$ on $\imum$ \cite{perelomov_casimir_1968}:
\begin{equation}
    \sum_{i_1 \ldots i_k} \tr{b_{i_1} \cdots b_{i_k}} O_{i_1} \cdots O_{i_k} \:, \qquad \mathrm{for} \:\:\: k = 1, 2,  \ldots m \:.
\end{equation}

\subsection{Higher order invariants}\label{sec: higher order invariants}

Recall that density matrices live in $\uM = \imum \oplus \imum^\perp$. To access information in the perpendicular component, we need to map $\rho$ outside the linear optical algebra. One option would be to compute the decomposition of $\uM=\bigoplus_{\ell=0}^n W_\ell$ into subspaces that are irreducible for $\mathrm{Ad}_{\varphi(S)}$ (see Section \ref{sec: irreducible representations}), but this requires computing the Clebsch-Gordan coefficients, which is computationally expensive \cite{wilkens_benchmarking_2024, arienzo_bosonic_2025}. With this decomposition, the projection of a density matrix into each subspace $W_\ell$ would have an invariant spectrum.

An alternative, inspired by Corollary \ref{corollary: trace invariants}, is to build invariant operators with certain products of Hamiltonians $O_i$ (related to higher moments \cite{rodari_observation_2025}). The number of different  operators $O_i$ in the product gives the order of the invariant. These products are related with anticommutators that are, in general, outside of the linear optical algebra, leading to new invariants.

\begin{prop}\label{prop: higher spectral invariant image}
For each $k\geq 1$, the spectrum of the Hermitian matrix
\begin{equation}\label{eq: higher spectral invariant image}
    P_k(\rho) := \sum_{i_1, \ldots, i_k}  {\rm tr}(O_{i_1} \cdots O_{i_k}  \rho) O_{i_1} \cdots O_{i_k}
\end{equation}
is invariant under linear optical evolution.
\end{prop}
\begin{proof}
    First, we prove that the matrix is Hermitian. Each term in \eqref{eq: higher spectral invariant image}, ${\rm tr}(O_{i_1} \cdots O_{i_k}  \rho) O_{i_1} \cdots O_{i_k}$, is summed with its Hermitian conjugate
    \begin{align*}
    \big( {\rm tr}(O_{i_1} \cdots O_{i_k}  \rho) O_{i_1} \cdots O_{i_k}\big)^\dagger 
    &= \overline{{\rm tr}(O_{i_1} \cdots O_{i_k}  \rho)} O_{i_k}^\dagger \cdots O_{i_1}^\dagger \\
    &= {\rm tr}(\rho^\dagger O_{i_k}^\dagger \cdots O_{i_1}^\dagger) O_{i_k} \cdots O_{i_1} \\
    &={\rm tr}(O_{i_k} \cdots O_{i_1} \rho) O_{i_k} \cdots O_{i_1}\:.
\end{align*}
Since the sum of a matrix and its Hermitian conjugate gives a Hermitian matrix, Eq. \eqref{eq: higher spectral invariant image} must be Hermitian.

To prove that the spectrum is invariant under linear optical evolution of $\rho$, apply the adjoint action \eqref{eq: adjoint action image} on each $O_i$:
\begin{align*}
    P_k(U\rho U^\dagger) &= \sum_{i_1, \ldots, i_k}  {\rm tr}(O_{i_1} \cdots O_{i_k}  U\rho U^\dagger) O_{i_1} \cdots O_{i_k} \\
    &= \sum_{i_1, \ldots, i_k}  {\rm tr}(U^\dagger O_{i_1} U \cdots U^\dagger O_{i_k} U \rho) O_{i_1} \cdots O_{i_k}  \\
    &=\sum_{i_1, \ldots, i_k} \sum_{j_1, \ldots, j_k}  C_{i_1j_1}\cdots C_{i_kj_k}{\rm tr}(O_{j_1}\cdots O_{j_k} \rho) O_{i_1} \cdots O_{i_k} \\
    &=\sum_{i_1, \ldots, i_k} {\rm tr}(O_{j_1}\cdots O_{j_k} \rho) U O_{i_1} U^\dagger \cdots U O_{i_k} U^\dagger \\
    &= U \left(\sum_{i_1, \ldots, i_k} {\rm tr}(O_{j_1}\cdots O_{j_k} \rho)  O_{i_1} \cdots O_{i_k} \right) U^\dagger = U P_k(\rho) U^\dagger \:.
\end{align*}
Therefore, the spectrum of $P_k(\rho)$ is invariant.
\end{proof}

Additionally, we show in Appendix \ref{appendix: invariant subspaces} how the eigenspaces of the operator $P_k$ relate to the invariant subspaces $W_\ell$ from Section \ref{sec: irreducible representations}.

In numerical experiments \cite{notebook}, we found that this invariant for $k=2$ is different for $\ket{N,N}$ and $(\ket{2N,0} + \ket{0, 2N})/\sqrt{2}$, allowing us to prove the impossibility of evolving one into the other (for a different proof using the covariance invariant, we refer to Section \ref{sec: example covariance}). On the contrary, the invariants from Theorem \ref{thm: spectrum projection} are the same for these two states and it is impossible to discard this state preparation with them. Thus, by mapping a state onto the perpendicular component of the algebra we gain access to additional information about linear optical evolution.

Another invariant involving higher moments is the spectrum of
\begin{equation}\label{eq: higher spectral invariant}
    \sum_{i_1, \ldots, i_k}  {\rm tr}(O_{i_1} \cdots O_{i_k}  \rho) b_{i_1} \cdots b_{i_k} \:, \qquad k\geq 2, \ldots, m \:.
\end{equation}
Following the proof in Corollary \ref{corollary: trace invariants}, the invariance of this spectrum can be used to recover a wide range of invariants that arise from the invariant tensor \eqref{eq: invariant tensor}. For example, by taking the trace of \eqref{eq: higher spectral invariant}, we obtain invariant quatities analogous to those in Corollary \ref{corollary: trace invariants}:
\begin{equation}\label{eq: higher order traces}
    \sum_{i_1,\ldots,i_k} {\rm tr}(b_{i_1}\cdots b_{i_k}) {\rm tr}( O_{i_1}\cdots{O}_{i_{k}} \rho)\:, \qquad k\geq 2, \ldots, m \:.
\end{equation}
Apparently, these invariants should have access to more information of the state than the tangent invariant (Theorem \ref{thm: spectrum projection} and Corollary \ref{corollary: trace invariants}) since they involve higher moments, which are not easy to measure experimentally. However, numerical experiments revealed that these invariants fail to detect the impossibility of preparing $(\ket{2N,0} + \ket{0, 2N})/\sqrt{2}$ from $\ket{N,N}$. We leave for future work understanding whether these invariants are able to detect more impossible linear optical evolutions than the invariants from Corollary \ref{corollary: trace invariants}.

Following the rule of contracting an invariant tensor with quantities that change under the adjoint action when $\rho$ is evolved with linear optics, one could construct new invariants in a similar way. Moreover, we could also construct other operators with an invariant spectrum under linear optical evolution of $\rho$, such as:
\begin{equation}
    \sum_{i_1\cdots i_k} {\rm tr}\left(b_{i_1}\cdots b_{i_k}\right)  [ O_{i_k},[\cdots[ O_{i_1},\rho]\cdots] \:.
\end{equation}

\subsection{Covariance invariant}

Another interesting second order invariant from a physical point of view arises from the covariance matrix, which also uses information about second order correlations:
\begin{theorem}
The spectrum of the covariance matrix
\begin{equation}\label{eq: covariance invariant}
    M(\rho)_{ij} = \langle O_i\rangle_\rho\langle O_j\rangle_\rho-\left\langle \frac{O_i O_j + O_j O_i}{2}\right\rangle_\rho
\end{equation}
is invariant under linear optical evolution.
\end{theorem}
\begin{proof}
To prove it, note that, using Eq. \eqref{eq: adjoint action image},
\begin{equation}
    {\rm tr }\left(O_i O_j U \rho U^\dagger\right)= 
    {\rm tr }\left(U^\dagger O_i U U^\dagger O_j U \rho\right)=
    \sum_{i'j'}^{m^2}C_{ii'}{\rm tr }\left( O_{i'} O_{j'} \rho\right) C^T_{j'j}
\end{equation}
and similarly for the other terms. Since $C$ is a rotation matrix, this implies that the spectrum of the covariance is, indeed, invariant:
\begin{equation}
    {\rm spec }\left(M(U \rho U^\dagger)\right) =
    {\rm spec }\left(C M(\rho)C^T\right) =
    {\rm spec }\left(M(\rho)\right) \:.
\end{equation}
\end{proof}

The proof shows that evolving a state under linear optics amounts to rotating the basis $\{O_i\}$ with a matrix $C$. Since the eigenvalues of the covariance matrix are the variances along the principal components, linear optics just rotates these principal components without changing their eigenvalues.

The terms of the covariance matrix are the expectation values of Hermitian operators, so they could, in theory, be measured experimentally. However, the implementation of successive measurements of $O_i$ and $O_j$ seems to require non-destructive measurements, which are challenging with current technology \cite{grangier_quantum_1998,nemoto_universal_2005}.

\section{Examples}\label{sec: examples}

In \cite{PGM23}, we used the orthonormalized version of the invariant $I_2(\rho)=\sum_{i=1}^{m^2}{\rm tr}(O_i\rho)^2$ to prove the impossibility of preparing a Bell state from a Fock state and the impossibility of preparing a $\ket{{\rm GHZ}}$ state from a $\ket{{\rm W}}$ state \cite{DVC00} (even if we add Fock ancillas in both cases). Now, we shall show applications of the new invariants.

\subsection{Tangent invariant of a Fock state}

As a first example, we prove that it is impossible to evolve one Fock state, $\ket{k_1, \ldots k_m}$, into another Fock state, $\ket{k'_1, \ldots k'_m}$, unless they are a permutation of one another. It is an alternative but similar proof to \cite{MRO14} using Theorem \ref{thm: spectrum projection}. If $\rho = \ket{k_1, \ldots k_m}\bra{k_1, \ldots k_m}$ is our density matrix, the spectrum of
\begin{equation}
h_\rho = \sum_{i=1}^{m^2}  {\rm tr}(O_i \rho) b_i = \sum_{j=1}^m  {\rm tr}(n_j \rho) b_j^z =
\begin{pmatrix}
\expect{n_1} & & \\
& \ddots &  \\
& & \expect{n_m}
\end{pmatrix}
=
\begin{pmatrix}
k_1 & & \\
& \ddots &  \\
& & k_m
\end{pmatrix}
\end{equation}
is invariant, where $b_j^z$ is given by Eq. \eqref{eq: basis algebra}. Therefore, the number of photons in each mode must be the same, except for a possible reordering of the modes.

\subsection{Tangent invariant of two coherent states}

One advantage of obtaining invariants with the adjoint action is that many invariants work regardless of whether we restrict the operators $O_i$ to subspaces of a fixed number of photons or not (because the matrix $C$ of the adjoint action in Eq. \eqref{eq: adjoint action image} is the same).

Let us put this in practice by applying, for example, the invariants in Corollary \ref{corollary: trace invariants} to two coherent states with amplitudes $\alpha$ and $\beta$, say $\rho = \ket{\alpha}\ket{\beta} \bra{\alpha}\bra{\beta}$. Since there are only two modes, Corollary \ref{corollary: trace invariants} gives two invariants:
\begin{align}
     & I_1(\rho) = \expect{{n}_1} + \expect{{n}_2} = |\alpha|^2 + |\beta|^2 \:, \\
     & I_2(\rho) = \expect{{n}_1}^2 + \expect{{n}_2}^2 + \left\langle{\frac{a^\dag_{1} a_2+ a^\dag_{2} a_1}{\sqrt{2}}}\right\rangle^2  
    + \left\langle{i\frac{a^\dag_{1} a_2- a^\dag_{2} a_1}{\sqrt{2}}}\right\rangle^2 = (|\alpha|^2 + |\beta|^2)^2 \:.
\end{align}
In this case, it happens that both invariants give the same information, the conservation of energy, but it gives a simple example of invariants applied to states in Fock space.

\subsection{Covariance invariant}\label{sec: example covariance}

The covariance invariant \eqref{eq: covariance invariant} can also be used to give an alternative proof to \cite{vanmeter_general_2007} and show that it is impossible to go from a Fock state $\ket{k, k'}$ to a NOON state $(\ket{N0}+\ket{0N})/\sqrt{2}$, where $k+k'=N>2$. The covariance matrices for $N>2$ are:
\begin{equation}
M(\ket{{\rm NOON}}) = -\frac{1}{4} \left(
    \begin{array}{cccc}
 N^2 & -N^2 & 0 & 0 \\
 -N^2 & N^2 & 0 & 0 \\
 0 & 0 & 2N & 0 \\
 0 & 0 & 0 & 2N \\
 \end{array}
 \right),
\end{equation}
\vspace{1mm}
\begin{equation}
M(\ket{k, k'})=-(2kk'+k+k')\left(
    \begin{array}{cccc}
 0 & 0 & 0 & 0 \\
 0 & 0 & 0 & 0 \\
 0 & 0 & 1 & 0 \\
 0 & 0 & 0 & 1 \\
 \end{array}
 \right).
\end{equation}
They clearly have different eigenvalues, so a transition from one to the other is impossible with passive linear optics. Note that for $N=2$ the matrices are different and they do have the same spectrum, as expected by the Hong-Ou-Mandel effect.

\subsection{Combining multiple invariants}

We can also combine multiple invariants to find conditions that trickier state preparations must fulfill. Let us take a coherent state plus a Fock state (one in each mode), $\ket{\psi_1}=\ket{\alpha}\ket{k}$, and a photon-added coherent state plus another Fock state, $\ket{\psi_2}={a}_1^\dagger\ket{\beta}\ket{k'}/\sqrt{1+|\beta|^2}$ (recall that photon-added coherent states are obtained by applying the creation operator to a coherent state \cite{zavatta_quantum--classical_2004}). The conservation of the mean number of photons gives a first condition: $|\alpha|^2 + k = \gamma + k'$, where $\gamma = (|\beta|^4 + 3|\beta|^2+1)/(1+|\beta|^2)$ is the mean number of photons of the photon-added coherent state. The conservation of the invariant $I_2(\rho)$ gives a second condition: $|\alpha|^4 + k^2 = \gamma^2 + k'^2$. These two equations imply that $|\alpha|^2 = k'$, $\gamma = k$ (if $k\neq k'$). These conditions are enough to prove the impossibility of preparing $\ket{\psi}_2$ from $\ket{\psi}_1$ deterministically when $k=1$ and $k'=0$.

More conditions arise from the conservation of the spectrum of the covariance matrix \eqref{eq: covariance invariant} of each state:

\begin{align}
M(\ket{\alpha}\ket{k}) = &\left(
    \begin{array}{cccc}
 -|\alpha|^2 & 0 & 0 & 0 \\
 0 & 0 & 0 & 0 \\
 0 & 0 & |\alpha|^2+k+2|\alpha|^2k & 0 \\
 0 & 0 & 0 & |\alpha|^2+k+2|\alpha|^2k \\
 \end{array}
 \right), \\[5mm]
M\left(\frac{{a}_1^\dagger\ket{\beta}\ket{k'}}{\sqrt{1+|\beta|^2}}\right)= \frac{1}{2} &\left(
\begin{array}{cccc}
 -4\frac{|\beta|^2+|\beta|^4+|\beta|^6}{(1+|\beta|^2)^2} & 0 & 0 & 0 \\
 0 & 0 & 0 & 0 \\
 0 & 0 & \gamma + k'\frac{3+7|\beta|^2+2|\beta|^4}{1+|\beta|^2} & 0 \\
 0 & 0 & 0 & \gamma + k'\frac{3+7|\beta|^2+2|\beta|^4}{1+|\beta|^2} \\
 \end{array}
 \right).
\end{align}
To alleviate the pain of computing the expectation values of the second state, we used the Julia package \texttt{QuantumAlgebra.jl} \cite{QuantumAlgebra.jl} to put the operators in normal order.

Three eigenvalues turn out to be equal, so we gain only an extra condition from this invariant: $k'=|\alpha|^2=2(|\beta|^2+|\beta|^4+|\beta|^6)/(1+|\beta|^2)^2$. And we could keep on calculating different invariants to obtain more conditions...

\section{Approximate state preparation}\label{sec: approximate state preparation}

Most photonic state preparations are impossible \cite{moyano-fernandez_linear_2017} because the dimension of Fock space grows combinatorially with the number of photons while the dimension of the linear optical unitaries only grows quadratically with the number of modes. This limits the usefulness of the invariants to discard impossible preparations. Instead, we usually want to prepare states using some heralded scheme with a high probability of preparing the target state. In this section we see how to relate the distance between states with the distance between their invariants. This gives a necessary condition for approximate state preparation with linear optics, including heralded preparations.

\begin{figure}[ht]
\centering
\includegraphics[scale=0.9]{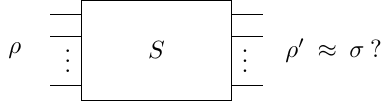}
\caption{Given an input state $\rho$ and a target $\sigma$, can we find a passive interferometer $S$ such that $\rho'=U\rho U^\dagger$ is close to $\sigma$?} \label{fig:approximate state preparation}
\end{figure}

Throughout this section, instead of working with the basis $\{O_i\}$ defined in \eqref{eq: basis image algebra}, we will work with an orthonormalized basis $\{O_i'\}$, so that the map
\begin{equation}
    \rho_T = \sum_{i=1} \tr{O_i' \rho} O_i'    
\end{equation}
is now an orthogonal projection onto the linear optical subalgebra \cite{PGM23}. To avoid excessive notation, we will also denote this projection with $\rho_T$. Working with orthogonal projections will allow us to use the Pythagorean theorem to bound the distance between the invariants.

\subsection{Distances between invariants}\label{distances}

Let $\rho$ be an input state, $\rho'$ the output state of the interferometer and $\sigma$ a target state that we would like to prepare from $\rho$. Let $\{\lambda^\rho_{Ti}\}$, $\{\lambda^{\rho'}_{Ti}\}$ and $\{\lambda^\sigma_{Ti}\}$ be the eigenvalues (in ascending order) of $\rho_T$, $\rho'_T$ and $\sigma_T$, respectively. From Theorem \ref{thm: spectrum projection}, we know that the invariants of the input and output states are equal: $\lambda^\rho_{Ti} = \lambda^{\rho'}_{Ti}$.

To bound the invariants, we use the Schatten norms for matrices $\|A\|_p=\tr{|A|^p}^{1/p}$. For $p=1$, it is the trace norm widely used in quantum information to measure distance between quantum states \cite{Nielsen_Chuang_2010} and, for $p=2$, it is the Frobenius norm $\tr{AA^\dagger}^{1/2}$. The distance between the input and target invariants can be bounded with the Hoffman-Wielandt inequality \cite{bhatia_analysis_1982}:
\begin{equation}
    d_T:=\left(\sum_{i} |\lambda^\rho_{Ti}-\lambda^\sigma_{Ti}|^2\right)^{1/2} = \left(\sum_{i} |\lambda^{\rho'}_{Ti}-\lambda^\sigma_{Ti}|^2\right)^{1/2} \leq \lVert \rho'_T - \sigma_T \rVert_2  \:.
\end{equation}
Now, because the norm $\|\cdot\|_2$ arises from an inner product, the Pythagorean Theorem allows us to decompose $\rho' - \sigma$ into two orthogonal components, one in $\imum$ and another in $\imum^\perp$:
\begin{equation}
    \lVert \rho' - \sigma \rVert_2^2 = \lVert \rho'_T - \sigma_T \rVert_2^2 + \lVert (\rho' - \rho_T) - (\sigma - \sigma_T) \rVert_2^2 \:.
\end{equation}
This equation implies that
$\lVert \rho'_T - \sigma_T \rVert_2 \leq \lVert \rho' - \sigma \rVert_2$. And, because of Schatten norm inequalities, this can be bounded again: $\lVert \rho' - \sigma \rVert_2 \leq \lVert \rho' - \sigma \rVert_1$. Therefore, the distance between the eigenvalues of $\rho_T$ and $\sigma_T$ can be bounded by the distance between $\rho'$ and $\sigma$:
\begin{equation}
\label{SpectralBound}
    d_T  \leq \lVert \rho'_T - \sigma_T \rVert_2 \leq \lVert \rho' - \sigma \rVert_2 \leq \lVert \rho' - \sigma \rVert_1 \:.
\end{equation}
This inequality gives a necessary condition for state preparation: if the eigenvalues of the input and target states are very different, the output and target states cannot be very close.

\subsection{Application to heralded state preparation}
These bounds have an immediate application to heralded state preparation. There are certain states, like Bell states, which can be used as a resource in many protocols, but can only be generated at low rates using parametric down-conversion or other complex processes. Instead, we would like to generate these advanced states from other easier-to-produce input states. As we have seen, most transformations are forbidden. Heralded state generation is a probabilistic way to produce those resource states \cite{forbes_heralded_2025}. Our bounds can be used to bound the probability of success of heralded generation.

Let $\ket{IN}$ be an input state we want to transform into a target state $\ket{T}$ using only linear optics. If this transformation is impossible, we can try heralded state generation, where we provide an evolution 
\begin{equation}
\label{eqH}U\ket{IN}\ket{A}\longrightarrow \ket{OUT}=\sqrt{p}\ket{T}\ket{H}+\sum_i \alpha_i \ket{\psi_i} \ket{H^i_\perp}
\end{equation}
for some ancillary modes in state $\ket{A}$ and success conditioned on measuring the state $\ket{H}$ in the ancillary modes. We have swept any phase of the target state under the $\alpha_i$ rug. In heralded state generation, we can guarantee success with probability $p$ as long as $\langle H \ket{H^i_\perp}=0$ for every $i$ and $\sum_i |\alpha_i|^2=1-p$. The ancillary input state, $\ket{A}$, and the herald, $\ket{H}$, are not necessarily the same state, but, as they are chosen to be easy to prepare or to measure; they will tend to be modes with either one or zero photons.

Now consider the known input state $\ket{IN}\ket{A}$, with density matrix $\rho$, and the desired output state $\ket{T}\ket{H}$, with density matrix $\sigma$. The actual ouput state $\rho'=\ket{OUT}\bra{OUT}$, as described in Eq. \eqref{eqH}, is unknown to us, but, if it has been generated from $\rho$ and linear optics, it must have the same invariant: $\mathrm{spec}(\rho'_T)=\mathrm{spec}(\rho_T)$. Similarly, for any valid output: $\tr{\sigma \rho'}=p$. For the pure states we are considering in heralded schemes
\begin{equation}
\label{BoundH}
\lVert\rho' - \sigma \rVert_2 =\sqrt{\tr{(\rho'-\sigma)^\dag (\rho'-\sigma})}=\sqrt{2}\sqrt{1-p} \:.
    \end{equation}
    
From the Hoffman-Wielandt inequality (Section \ref{distances}), we know that:
\begin{equation}
    d_T:=\left(\sum_{i=1}^\infty |\lambda^\rho_{Ti}-\lambda^\sigma_{Ti}|^2\right)^{1/2} \leq \lVert \rho'_T - \sigma_T \rVert_2  \:.
\end{equation}
Likewise, for the complement,
\begin{equation}
    d_\perp:=\left(\sum_{i=1}^\infty |\lambda^\rho_{\perp i}-\lambda^\sigma_{\perp i}|^2\right)^{1/2} \leq \lVert (\rho'-\rho'_T) - (\sigma-\sigma_T) \rVert_2  \:,
\end{equation}
and we have
\begin{equation}
       d_T^2\leq \lVert \rho'_T - \sigma_T \rVert_2^2 \:, \quad d_\perp^2 \leq \lVert (\rho'-\rho'_T) - (\sigma-\sigma_T) \rVert_2^2  \:.
\end{equation}
We can use the Pythagorean theorem for the Frobenius norm on orthogonal matrices and Eq. \eqref{BoundH} to show:
\begin{equation}
d_T^2+d_\perp^2\leq ||\rho'-\sigma||_2^2=2(1-p) \:,
\end{equation}
which gives a probability bound for the heralded operation
\begin{equation}\label{eq: probability bound}
p\leq 1-\frac{d_T^2+d_\perp^2}{2} \:.
\end{equation}

Finally, note that we could apply the same argument to the projections of $\rho$ and $\sigma$ onto each of the invariant subspaces from Section \ref{sec: irreducible representations}, $\{W_\ell\}_{\ell=0}^n$, using orthonormal bases for $W_\ell$. Computing the distance between the spectra in each $W_\ell$ and applying Pythagoras, we could give tighter bounds to inequality \eqref{eq: probability bound}.

\subsection{Example: preparation of \texorpdfstring{$\ket{20}$}{|20>} from \texorpdfstring{$\ket{11}$}{|11>}}\label{appendix: heralded example}

We can use these results to bound the maximum probability for the forbidden transition $\ket{11}\to \ket{20}$ conditioned on finding no photons in the second mode. This can be seen as a heralded transition from $\ket{1}$ to $\ket{2}$ with one ancillary mode with one photon. Therefore, our input density matrix is $\rho = \ket{11}\bra{11}$ and the target density matrix is $\sigma=\ket{20}\bra{20}$.

The first step in computing the bound \eqref{eq: probability bound} is to precompute an orthonormal basis for the algebra $\imum$, which is done by orthonormalizing the basis \eqref{eq: basis image algebra} with our software package QOptCraft \cite{notebook}. Then, we diagonalize $\rho_T$ and $\rho_\perp$ and compute their spectra: $\spec{\rho_T}=(\frac{1}{3},\frac{1}{3},\frac{1}{3})$ and $\spec{\rho_\perp} = (-\frac{1}{3},-\frac{1}{3},\frac{2}{3})$. Any output state $\rho'$ that we can reach has those spectra in the corresponding subspaces. Next, we compute the projections of $\sigma$ and their spectra: $\spec{\sigma_T}=(-\frac{1}{6},\frac{1}{3},\frac{5}{6})$ and $\spec{\sigma_\perp}=(-\frac{1}{3},\frac{1}{6},\frac{1}{6})$. Finally, we compute
\begin{equation}
d_T^2=\left|\frac{1}{3}-\frac{-1}{6}\right|^2+\left|\frac{1}{3}-\frac{1}{3}\right|^2+\left|\frac{1}{3}-\frac{5}{6}\right|^2=\frac{1}{4}+0+\frac{1}{4}=\frac{1}{2} \:,
\end{equation}
and
\begin{equation}
d_\perp^2=\left|\frac{-1}{3}-\frac{-1}{3}\right|^2+\left|\frac{-1}{3}-\frac{1}{6}\right|^2+\left|\frac{2}{3}-\frac{1}{6}\right|^2=0+\frac{1}{4}+\frac{1}{4}=\frac{1}{2} \:,
\end{equation}
which, substituting in \eqref{eq: probability bound}, bounds the success probability of the transition by $p\leq \frac{1}{2}$. This bound is tight since a balanced beamsplitter produces an output $(\ket{20}-\ket{02})/\sqrt{2}$.

However, this decomposition into two subspaces doesn't give tight bounds for other heralded generation of entangled states, like NOON or Bell states. We leave for future work to refine these bounds with the decomposition $\uM = \bigoplus_{\ell=0}^n W_\ell$ from Section \ref{sec: irreducible representations}.

\section{Software}\label{sec: code}

Even though analytical computations are useful for proving theorems, they are complicated for most states and our invariants are better suited for numerical calculations. We implemented the invariants in our Python library QOptCraft \cite{aguado_qoptcraft_2023}, with an example notebook accompanying this article \cite{notebook}.

In order to make the computations more efficient, the bases of the linear optical algebras are saved as sparse matrices. This is especially important as the number of photons grows. Precomputing and saving the basis is also a good idea as computing it from scratch every time increases the computational time of finding the invariants.

Finally, we note that currently there is only the possibility of using states with a fixed number of photons. Coherent or squeezed states are not implemented yet.

\section{Discussion and outlook}\label{sec: discussion}

Preparing useful quantum states with linear optics remains a fundamental challenge in quantum information. Our work shines a light on the mathematical structure of linear optics, giving a recipe to define invariant quantities that must be conserved by passive linear optical evolution. These invariants are especially useful when trying to prepare, exactly or approximately, some state from another one (presumably, easy to generate). They give necessary conditions for this preparation to be possible. These conditions could be implemented in some of the numerous algorithms for computer-aided state preparation \cite{knott_search_2016, gao_computer-inspired_2020, krenn_automated_2016} to discard impossible preparations.

Our invariants have the advantage that the quantum state is not restricted to a fixed number of photons (i.e. a finite-dimensional space): we can work with the full Fock space. This improves other invariants in the literature that can only be applied to finite-dimensional spaces \cite{somma_quantum_2005, MRO14, PGM23}. A second improvement is that they can be used with mixed states, not only pure ones, like the invariant in \cite{MRO14}. Another good property of some of our invariants (the ones involving expectation values of the operators $O_i$) is that they can be measured experimentally \cite{NK06}. In fact, they were recently measured in an experimental setup by two independent groups \cite{rodari_observation_2025, yang_experimental_2025}. Moreover, in \cite{rodari_observation_2025}, Rodari \textit{et al.} relate the tangent invariant with conserved quantities of the coherency matrix of first-order correlations and, in \cite{yang_experimental_2025}, Yang \textit{et al.} relate it to the block-diagonal structure of the Hermitian transfer matrix. In \cite{rodari_observation_2025}, the authors also propose using the invariants to benchmark boson sampling experiments and to certify non-linear optical evolutions, ideas that we would like to look into in future work. Lastly, Thomas \textit{et al.} showed recently \cite{thomas_shedding_2025} that the tangent invariants can be efficiently estimated with any desired precision from classical shadows: random measurements on a few copies of the quantum state that can be used to predict properties of that state.

One could wonder: what about sufficient conditions for state preparation beyond SU(2) evolutions? While we have not addressed it in this article, in future works we plan to investigate whether we can provide a sufficient condition by combining several invariants, restricting more and more the allowed linear optical evolutions.

Another future line of research is to use the invariants to bound the success probabilities of interesting heralded state preparations, like Bell state preparation. This could, in turn, be used to explore the limits of heralded and post-selected optical quantum gates \cite{knill_quantum_2002, stanisic_generating_2017}.

Finally, one could ask if this approach can be applied to active linear optics. It turns out it cannot, since the map between the quasiunitary matrices describing active linear optics \cite{leonhardt_explicit_2003} and the unitary matrices describing their quantized version is not a group homomorphism. And we really need the map $\varphi$ to be a homomorphism to derive Equation \eqref{eq: adjoint action image}. Nonetheless, in future research, we will try to circumvent this issue by decomposing the active linear optical system into a passive one with squeezing in between.

\section*{Acknowledgements}

The authors would like to thank the anonymous referees for their thorough review of the manuscript. Their insightful comments greatly improved the clarity and precision of the article.

P.~V.~P. has been funded under the UVa 2024 predoctoral contract, co-financed by Banco Santander. P.~V.~P. and J.C.~G.-E. were supported by the European Union.-Next Generation UE/MICIU/Plan de Recuperaci\'on, Transformaci\'on y Resiliencia/Junta de Castilla y Le\'on, and the Department of Education, Junta de Castilla y León, and FEDER Funds (reference: CLU-2023-1-05).
V.~G. has been supported in part  by the project AICO/2023/035 funded by Conselleria de Educaci\'o, Cultura, Universitats i Ocupaci\'o. J.~J.~M.~F. was partially supported by grants PID2022-138906NB-C22 funded by MCIN/AEI/10.13039/501100011033 (Spain), and by the ``European Union NextGenerationEU/PRTR'', as well as by Universitat Jaume I, Spain, grants UJI-B2021-02, GACUJIMA/2023/06 and GACUJIMB/2023/03. In addition, J.C.~G.-E. received funding from the Spanish Government and FEDER grant PID2020-119418GB-I00 (MICINN) and FEDER and Junta de Castilla y Le\'on  VA184P24.





\printbibliography

\newpage

\appendix

\section{Invariant subspaces of equivariant operators}\label{appendix: invariant subspaces}

We begin this appendix by refreshing some concepts from representation theory \cite{hall_lie_2015,ragone_representation_2023}. Let $G$ be a Lie group and $V$ a vector space called the \textit{carrier space}. A representation $R:G \rightarrow  \mathrm{Aut}(V)$ is a group homomorphism from $G$ to the group of automorphisms in $V$. A subspace $W\subset V$ is said to be \textit{invariant} if $R(g)W\subset W$ for all $g\in G$. If a representation doesn't have any proper non-trivial invariant subspaces, it is called \textit{irreducible}. We will call a subspace irreducible if the subrepresentation induced in that subspace is irreducible.

A map $T: V \rightarrow V$ is said to be \textit{equivariant} if
\begin{equation}
    T(R(g)v) = R(g)T(v) \:, \quad\: \text{for all} \:\: v\in V\:, \:\: g\in G \:.
\end{equation}
If an equivariant operator is also linear and self-adjoint, its eigenspaces will be invariant under the representation: let $v_\lambda$ be an eigenvector with eigenvalue $\lambda$,
\begin{equation}
    T(R(g)v_\lambda) = R(g)T(v_\lambda) = R(g) \lambda v_\lambda \:.
\end{equation}
Hence, $R(g)v_\lambda$ is also an eigenvector with eigenvalue $\lambda$.

In fact, we can characterize any linear and self-adjoint equivariant operator, $T$, by decomposing $V$ into irreducible subspaces: $V=\bigoplus_{\ell=0}^n W_\ell$. By Schur's lemma \cite[Theorem 4.29]{hall_lie_2015}, we know that any equivariant linear operator in an irreducible subspace must be a multiple of the identity. In our case, if $\Pi_\ell$ is the projector to $W_\ell$, the restriction of $T: V \rightarrow V$ to $W_\ell$ is
\begin{equation}
    T_\ell := \Pi_\ell \circ T \circ \Pi_\ell = \alpha_\ell \,\mathcal{I}_{W_\ell}  \:,
\end{equation}
for some $\alpha_\ell \in \mathbb{C}$ (by Schur's lemma). In fact, $\alpha_\ell \in \mathbb{R}$ because $T$ is self-adjoint. Therefore, any linear, self-adjoint and equivariant operator must be a linear combination of projectors onto irreducible subspaces
\begin{equation}
    T = \sum_{\ell=0}^n \alpha_\ell \Pi_\ell \:, \qquad \alpha_\ell \in \mathbb{R} \:.
\end{equation}
Consequently, any eigenspace of $T$ must be a direct sum of irreducible subspaces.

\subsection*{Example}

The previous argument provides a recipe for computing invariant subspaces (not necessarily irreducible) for the linear optical representation \eqref{eq: linear optical ad rep}, $\mathrm{Ad}_{\varphi(S)}$. It is an alternative to computing the decomposition into irreducible subspaces (see Section \ref{sec: irreducible representations}) via the Clebsch-Gordan coefficients \cite{arienzo_bosonic_2025, wilkens_benchmarking_2024}.

Recall that we already saw an equivariant linear operator for $\mathrm{Ad}_{\varphi(S)}$ in Proposition \ref{prop: higher spectral invariant image}:\\ $P_k: \uM \rightarrow \uM$ given by
\begin{equation*}
        P_k(\rho) = \sum_{i_1, \ldots, i_k}  {\rm tr}(O_{i_1} \cdots O_{i_k}  \rho) O_{i_1} \cdots O_{i_k} \:.
\end{equation*}
We can check that it is also self-adjoint for the Hilbert-Schmidt inner product:
\begin{equation}
    \begin{aligned}
    \langle \sigma, P_k(\rho)\rangle = &\left\langle \sigma, \sum_{i_1, \ldots, i_k} \tr{ O_{i_1} \cdots O_{i_k}\rho} O_{i_1} \cdots O_{i_k} \right\rangle  = \\[2mm]
    &\sum_{i_1, \ldots, i_k} \tr{\sigma O_{i_1} \cdots O_{i_k}} {\rm tr}(O_{i_1} \cdots O_{i_k} \rho) = \\[1mm]
    &\left\langle \sum_{i_1, \ldots, i_k} \tr{ O_{i_1} \cdots O_{i_k}\sigma} O_{i_1} \cdots O_{i_k}, \rho \right\rangle = \langle P_k(\sigma), \rho\rangle \:.
    \end{aligned}
\end{equation}
Therefore, $\rho$ can be projected onto each eigenspace, $V_\lambda$, of $P_k$:
\begin{equation}
    \Pi_\lambda \rho = \sum_i \tr{v_i^\lambda \rho} v_i^\lambda \:,
\end{equation}
where $\Pi_\lambda$ is the projector onto $V_\lambda$ and $\{v_i^\lambda\}$ is an orthonormal eigenbasis of $V_\lambda$. Since the subspace is invariant,
\begin{equation}
    \Pi_\lambda\mathrm{Ad}_{\varphi(S)}(\rho) =
    \mathrm{Ad}_{\varphi(S)}(\Pi_\lambda\rho) \:,
\end{equation}
we conclude that the spectrum of $\Pi_\lambda\rho$ is an invariant under linear optical evolution.

In \cite{notebook}, we computed numerically the eigenspaces of $P_k$ and showed the linear optical invariance of the spectrum of each $\Pi_\lambda\rho$. This set of invariants would detect, at least, all forbidden evolutions detectable by the spectrum of $P_k$. In fact, numerical experiments show that there are $n+1$ invariant eigenspaces of $P_k$, which means these are exactly the $n+1$ irreducible subspaces for the $n$-photon representation $\mathrm{Ad}_{\varphi(S)}$, providing an alternative means for computing them without the Clebsch-Gordan coefficients (but not necessarily faster). This explains why the spectrum of $P_k(\rho)$ for $k=2$ is able to discard impossible state preparations that the spectrum of $P_1(\rho)=\rho_T\in\imum$ is unable to: because it is also accessing information in the perpendicular subspace $\imum^\perp$.

\end{document}